\documentclass[10pt,twocolumn,twoside]{IEEEtran}
\usepackage{cite}
\usepackage{amsmath,amssymb,amsfonts}
\usepackage{graphicx}
\usepackage{textcomp}
\usepackage{mathrsfs} 

\usepackage{tabulary}
\usepackage{blindtext}
\usepackage[font=footnotesize]{caption}
\usepackage{subfig}
\usepackage{amsthm}  
\usepackage{dsfont}
\usepackage{algorithm2e}
\RestyleAlgo{ruled}
\usepackage{algpseudocode}
\usepackage{tabularx}
\usepackage{float}
\usepackage{bbm} 

\newcommand{\vect}[1]{\boldsymbol{#1}}
\newcommand{\mat}[1]{\boldsymbol{#1}}

\newcommand{\wh}[1]

\DeclareMathOperator{\diag}{\text{diag}}
\DeclareMathOperator*{\argmax}{\rm{argmax}} 
\renewcommand{\eqref}[1]{Eq.~(\ref{#1})}  

\newtheorem{remark}{Remark}
\newtheorem{theorem}{Theorem}
\newtheorem{lemma}{Lemma}
\newtheorem{assumption}{Assumption}

\newtheorem{corollary}{Corollary}
\newtheorem{proposition}{Proposition}

\newtheorem{problem}{Problem}

\allowdisplaybreaks
\def\BibTeX{{\rm B\kern-.05em{\sc i\kern-.025em b}\kern-.08em
    T\kern-.1667em\lower.7ex\hbox{E}\kern-.125emX}}
\begin{document}
\title{
Controlling a Social Network of Individuals with Coevolving Actions and Opinions}

\author{Roberta Raineri, Mengbin Ye, and Lorenzo Zino
\thanks{R. Raineri and  L. Zino are with the Department of Electronics and Telecommunications, Politecnico di Torino, Turin, Italy (\texttt{\{roberta.raineri,lorenzo.zino\}@ polito.it}).  M. Ye is with the School of Computer and Mathematical Sciences, University of Adelaide, Adelaide, Australia (\texttt{ben.ye@adelaide.edu.au}). 
This work was partially supported by the Western Australian Government through the Premier's Science Fellowship Program.  }%
}

\maketitle

\begin{abstract}
In this paper, we consider a population of individuals who have actions and opinions, which coevolve, mutually influencing one another on a complex network structure. In particular, we formulate a control problem for this social network, in which we assume that we can inject into the network a committed minority ---a set of stubborn nodes--- with the objective of steering the population, initially at a consensus, to a different consensus state. Our study focuses on two main objectives: i) determining the conditions under which the committed minority succeeds in its goal, and ii) identifying the optimal placement for such a committed minority. After deriving general monotone convergence result for the controlled dynamics, we leverage these results to build a computationally-efficient algorithm to solve the first problem and an effective heuristics for the second problem, which we prove to be NP-complete. For both algorithms, we establish theoretical guarantees. The proposed methodology is illustrated though academic examples, and demonstrated on a real-world case study. 
\end{abstract}

\section{Introduction}\label{sec:intro}

Over the past decades, the systems and control community have witnessed a growing interest in developing and analyzing mathematical models to study, forecast, and control complex social phenomena and collective human behavior~\cite{friedkin2015_socialsurvey,proskurnikov2017tutorial,proskurnikov2018tutorial_2,Fontan2018,DePasquale2022,Bizyaeva2023,Bernardo2024}. Within this general effort, particular interest has been devoted to collective decision-making, whereby a population of individuals have to repeatedly make decisions on a specific action to take (often binary) on the basis of several factors, including their opinions on the considered action. For instance, this scenario often arises in different contexts of social change problems: individuals may decide whether to use a disposable cup or a reusable cup to have a coffee, or whether to use inclusive language or not when writing an email. In these contexts, empirical evidence and social psychology theories suggest that decision-making is deeply intertwined with opinion formation processes~\cite{gavrilets2017collective,lindstrom2018role}. This calls for the development of model paradigms able to integrate opinion formation processes within the decision-making model in a coevolutionary fashion.

The continuous-opinion discrete-action model is a first step in this direction. The seminal paper by Martins~\cite{Martins2008coda} and its main extensions~\cite{Ceragioli2018quantized,Tang2021coda} built on classical opinion dynamics models~\cite{proskurnikov2017tutorial}, and rely on the assumption that the opinion formation process entirely shapes the decision-making, whereby actions are a quantization of opinions. Despite relevant for many applications, this assumption limits the possibility to capture the presence of a misalignment between individuals' personal opinions and their actions, which is often observed in real-world social systems. This is the case, e.g., of the phenomenon of unpopular norms~\cite{centola2005emperor,willer2009false_enforcement}, whereby a community keeps exhibiting a collective behavior that is disapproved by the most of its members. This limitation was addressed in~\cite{zino2020chaos,Hassan2023tac}, where a coevolutionary model of actions and opinions was proposed. This model is built by incorporating an opinion formation process within a game-theoretic framework used to model decision-making~\cite{montanari2010spread_innovation,Jackson2015}, and allows individuals to simultaneously revise their (binary) actions and share their opinions on their support for the action on a complex network, accounting for social pressure, opinion influence, and self-consistency. The analysis of this model has shown its ability to reproduce several real-life phenomena, including the emergence and persistence of unpopular norms and polarization.

In this paper, we take a step further from the analysis of social systems to their control. In the literature, the two most common approaches for controlling social systems are to assume that one can either i) provide monetary or societal incentives to favor a desired opinion/action over the others~\cite{guo2013algebraic,Riehl2016,riehl2018incentive,Quijano2017,bacsar2024inducement} or ii) directly control a committed minority of stubborn agents. Here, we focus on the second type of control action, which has been extensively studied in the context of opinion dynamics~\cite{ghaderi2014opinion,Yuhao2021,Wang2023} and decision-making~\cite{kempe2003_maxspread,montanari2010spread_innovation,centola2018experimental_tipping,ye2021nat,Como2022supermodular}, but it is still unexplored for the coevolutionary model.

Building on the model proposed in~\cite{Hassan2023tac}, we focus on the control problem of unlocking a paradigm shift in a population by means of a committed minority. Specifically, we consider a population initially at a consensus in which all individuals select and support the same action. Then, we introduce a committed minority of stubborn agents with the goal of steering the entire population to a consensus on the opposite action. Stubborn agents may consistently select the opposite action, share opinions supporting it, or both. This problem is relevant to many real-life applications. For instance, unlocking a paradigm shift is key for social change, e.g., to favor the collective transition towards more sustainable practices. From the opposite perspective, understanding how a committed minority can steer an entire population to a desired collective behavior is key to guaranteeing robustness of social systems against malicious attacks~\cite{Schneider2011}.

By leveraging systems and control theoretic tools, we study the controlled model and we establish an array of properties and theoretical results, including monotonic convergence to an equilibrium point. Building on these theoretical findings, we establish necessary and sufficient conditions for the set of stubborn nodes to unlock a paradigm shift, depending on the model parameters and on the network structure. Then, we deal with the problem of identifying the minimal committed minority needed to control the network. After demonstrating analytically that the problem is NP-complete, we use our theoretical results and we get inspiration from~\cite{Como2022supermodular} to design an efficient iterative algorithm for its solution.

In details, the main contribution of this paper is four-fold. First, we incorporate a control action in the coevolutionary model~\cite{Hassan2023tac} and we formulate two control problems: determining if a control action is sufficient to steer the population to the desired consensus (\emph{effectiveness guarantee problem}), and identifying the minimal set of nodes to be controlled to achieve such a goal (\emph{minimal control set identification problem}). Second, we prove an array of general properties for the controlled dynamics, including convergence, and we characterize the complexity of our research problems, demonstrating that the minimal control set identification problem is NP-complete, which limits the possibility to adopt classical heuristics to approximate its solution. Third, we propose an algorithm to solve the effectiveness guarantee problem and, after prove its effectiveness analytically, we use it to evaluate the impact of the model parameters on a synthetic case study, where a closed-form solution of the problem can be derived. Fourth, we propose an iterative algorithm for the minimal control set identification problem with probabilistic convergence properties, and we demonstrate its efficiency on a case study, whose network topology is reconstructed from real-world face-to-face contact data~\cite{malawi_net}. 

Some of the results of this paper appeared, in a preliminary form, in~\cite{raineri2024}. Here, besides expanding the introduction to frame our contribution within the related literature and applications, we extend our preliminary results along several lines. First, we generalize the control action, allowing only part of the state of agents to be controlled. Second, we demonstrate that the minimal control set identification problem is NP-complete and it is not sub-modular (Theorem~\ref{th:np} and Proposition~\ref{prop:sub}). Third, we extend the algorithm proposed in~\cite{raineri2024} to solve the effectiveness guarantee problem to the more general control setting considered in this paper, and we prove its effectiveness in Theorem~\ref{th:algorithm}, which was missing in our previous publication. Fourth, we propose a novel algorithm to solve the minimal control set identification problem for general networks and, after proving its effectiveness in Theorem~\ref{th:convergence_MC}, we demonstrate it on a realistic case study.

The rest of the paper is organized as follows. In Section~\ref{sec:model}, we introduce the controlled coevolutionary dynamics and we formulate our two research problems. In Section~\ref{sec:results}, we present some general results on the uncontrolled and controlled coevolutionary dynamics. Sections~\ref{sec:problem1} and~\ref{sec:problem2} are devoted to the solution of the effectiveness guarantees and the minimal control set problems, respectively. In Section~\ref{sec:case}, we present a numerical case study. Section~\ref{sec:conclusion} concludes the paper.

\section{Model and Problem Statement}\label{sec:model}

\textit{Notation}.  We denote a vector $\vect{x}$ with bold lowercase font, with $x_{i}$ its $i$th entry; and a matrix $\mat{A}$ with bold capital font, and $a_{ij}$ the $j$th entry of its $i$th row. 
The all-$1$ column vector is denoted as $\vect{1}$, with appropriate dimension depending on the context. Given two vectors $\vect{x},\vect{y}$ with same dimension, we use  $\vect{x}\leq\vect{y}$ to denote $x_i\leq y_i$, for all entries $i$.

\subsection{(Uncontrolled) Coevolutionary Model}
\label{sec:dynamics}
We consider a population  $\mathcal V = \{1,\hdots,n\}$ of $n$ individuals. Each  $i\in\mathcal V$ is associated with a two-dimensional state variable $(x_i(t),y_i(t))\in\{-1,+1\}\times[-1,+1]$, with discrete time $t$: $x_i(t)\in\{-1,+1\}$ represents the \emph{action} of individual $i$ at time $t$, $y_i(t)\in[-1,+1]$ their \emph{opinion} on the action ($y_i(t)=-1$ means that $i$ is totally in favor of action $-1$,  $y_i(t)=+1$ that $i$ fully supports action $+1$). Actions and opinions are gathered in vectors $\vect{x}(t)\in\{-1,1\}^n$ and  $\vect{y}(t)\in[-1,1]^n$, and the state of the system is fully represented by the joint $2n$-dimensional vector $\vect{z}(t):=(\vect{x}(t),\vect{y}(t))\in\{-1,1\}^n\times [-1,1]^n$. Given an individual $i\in\mathcal V$, we define as $\vect{z_{-i}}:=(\vect{x_{-i}},\vect{y_{-i}})\in\{-1,1\}^{n-1}\times [-1,1]^{n-1}$ the $(2n-2)$-dimensional vector with the state of all other individuals. At each time step $t$, we define a set $\mathcal R(t)\subseteq \mathcal V$, and we assume that all individuals in this set  simultaneously revise their state at time $t$. 

\begin{assumption}[Revision sequence]\label{a:activation}
   There exists a constant $T<\infty$ such that $\cup_{s=0}^{T-1} \mathcal R(t+s)=\mathcal V$, for any $t\geq 0$.
\end{assumption}

\begin{remark}\label{rem:act}
Assumption~\ref{a:activation} encompasses many synchronous and asynchronous update rules: for synchronous update rules, $\mathcal R(t)=\mathcal V$ for all $t$; for asynchronous update rules, $\mathcal R(t)$ comprises a single individual. The requirement $\cup_{s=0}^{T-1} \mathcal R(t+s)=\mathcal V$ closely mirrors typical assumptions in opinion dynamics with time-varying networks~\cite[Section~3]{proskurnikov2018tutorial_2}, as it ensures that over $T$ consecutive time steps, every agent activates at least once. Note that no restrictions are imposed on the ordering of the agent activations, allowing for both deterministic or stochastic mechanisms to determine the activation sequence. 
\end{remark}

At time $t$, each individual $i\in\mathcal R(t)$ updates their state, aiming to maximize the utility function defined in~\cite{Hassan2023tac}, that accounts for three contributions: 
 i) individuals' tendency to coordinate actions;
  ii) opinions exchanged with peers; and
iii) an individual's tendency to have consistently between their action and opinion. Following~\cite{Hassan2023tac}, we define the utility that $i$ receives for selecting an  action and opinion pair $\vect{z_i}=(x_i,y_i)$ when the state of the others is $\vect{z_{-i}}$ as
\begin{equation}\label{eq:utility}
{{
\begin{array}{l}
    u_i(\vect{z_i},\vect{z_{-i}})=\frac{\lambda_i (1-\beta_i)}{2}\displaystyle\sum_{j \in \mathcal V} a_{ij} \big[ (1-x_j)(1-x_i + (1+x_j)\cdot\\(1+x_i) \big]  -\beta_i (1-\lambda_i)\displaystyle \sum_{j \in \mathcal V}{w}_{ij}(y_i-y_j)^2\hspace{-.1cm} -\hspace{-.1cm} \lambda_i\beta_i(x_i-y_i)^2\hspace{-.1cm},
\end{array}}}
\end{equation}
where $a_{ij} \in [0,1]$ and $w_{ij} \in [0,1]$ are the influence of individual $j$'s action and opinion, respectively; and $\lambda_i \in (0,1]$ and $\beta_i\in(0,1]$ the weights given to actions observed and opinions exchanged, respectively. The quantities $a_{ij}$ and $w_{ij}$ are gathered into two matrices $\mat A$ and $\mat W$, which we assume to be stochastic (i.e., $\mat{A}\vect{1}=\mat{W}\vect{1}=\vect{1}$). Such a structure induce a two-layer network $\mathcal  G = (\mathcal V, \mathcal E_A, \mat{A}, \mathcal E_W, \mat{W})$, where $\mathcal E_A$ are the edges on the influence layer on which individuals see others' actions and $\mathcal E_A$ are the edges on the communication layer, on which individuals discuss about their opinions. All parameters are summarized in Table~\ref{tab:parameters}. Before explicitly presenting the uncontrolled coevolutionary dynamics, we make some  observations on \eqref{eq:utility}, with proof  reported in Appendix~\ref{sec:proof_super}.


\begin{proposition}\label{prop:super}
    A game with the utility function in \eqref{eq:utility} is supermodular. 
\end{proposition}

\begin{table}    \caption{Models variables and parameters. }
    \label{tab:parameters}
    \centering
\begin{tabular}{r|l}
$x_i(t)\in\{-1,+1\}$&action of individual $i$ at time $t$\\
$y_i(t)\in[-1,+1]$&opinion of individual $i$ at time $t$\\
$a_{ij}\in[0,1]$&influence of $j$'s action on $i$\\
$w_{ij}\in[0,1]$&influence of $j$'s opinion on $i$\\
$\lambda_i\in(0,1]$&weight of actions\\
$\beta_i\in(0,1]$&weight of opinions\\
    \end{tabular}
\end{table}

In \eqref{eq:utility}, we enforce $\lambda_i>0$ and $\beta_i>0$ to guarantee a nontrivial coupling between the two variables. In the limit case in which one of these parameters is equal to $0$, the coevolutionary model would reduce to a simpler (and well-known) dynamics, as commented in the following.

\begin{remark}
   The utility in \eqref{eq:utility} generalizes classical network majority (and coordination) games~\cite{montanari2010spread_innovation,Jackson2015} (obtained for  $\lambda_i=1$) and the French-DeGroot opinion dynamics model~\cite{marden2009game,proskurnikov2017tutorial} (for $\beta_i=1$). See,~\cite{Hassan2023tac} for more details.
\end{remark}

We are now ready to present the coevolutionary dynamics, in which agents who activate seek to maximize their utility function in \eqref{eq:utility}. Consequently, for each $i\in\mathcal V$, the action and opinion are revised as follows:
\begin{equation}\label{eq:br}x_i(t+1),y_i(t+1)\hspace{-.1cm}=\hspace{-.1cm}\left\{\hspace{-.1cm}\begin{array}{ll}\displaystyle\argmax_{\vect{z_i}\in\{-1,1\}\times[-1,1]}\hspace{-.4cm}{u_i(\vect{z_i},\vect{z_{-i}})} &i\in\mathcal R(t),\\
x_i(t),y_i(t)&i\notin \mathcal R(t),\end{array}\right.
\end{equation}
with the convention that, when the $\argmax u_i(\vect{z_i},\vect{z_{-i}})$ comprises multiple elements, we set $x_i(t+1)=x_i(t)$. In other words, each individual $i\in\mathcal R(t)$ performs a  joint best-response with respect to \eqref{eq:utility}, as illustrated in Fig.~\ref{fig:schema}.

\begin{figure}
    \centering
\includegraphics{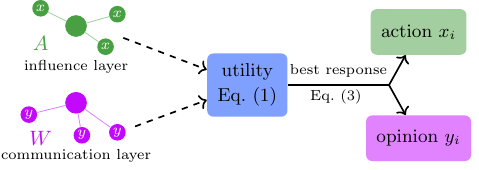}   \caption{Schematic of the update rule for a generic individual $i\in\mathcal R(t)$. }
    \label{fig:schema}
\end{figure}

In Proposition~\ref{prop:super}, we proved that the game is supermodular, which yields some important properties that can be used to prove convergence of \eqref{eq:br}.  However, to characterize its equilibria, we need to derive a closed-form expression for \eqref{eq:br}, following~\cite{Hassan2023tac,raineri2024}.

\begin{proposition}\label{prop:dynamics}
Individual $i\in\mathcal R(t)$ updates their state as:
\begin{subequations}\label{eq:dinamics}
\begin{align}
\label{x-dinamic}
&x_i(t+1) = s(\vect z(t)),\\
\label{y-dinamic}
&y_i(t+1)=(1-\lambda_i)\sum\nolimits_{j \in \mathcal V} {w}_{ij}y_j(t) + \lambda_i s(\vect{z}(t)),
\end{align}
\end{subequations}
where 
$$s(\vect z(t))=\begin{cases}
    +1 \qquad &\text{if}\; \delta_i(\vect z(t)) >0, \\
    -1 \qquad &\text{if}\; \delta_i(\vect z(t)) <0, \\
    x_i(t) \qquad &\text{if}\; \delta_i(\vect z(t)) =0,
\end{cases}
$$
with
$$
\delta_i(\vect z(t)) = 2 \beta_i (1-\lambda_i)\sum_{j \in \mathcal V} {w}_{ij} y_j(t) + (1- 
 \beta_i) \sum_{j \in \mathcal V} a_{ij} x_j(t).
$$
\end{proposition}

From Proposition~\ref{prop:dynamics}, we derive the following observation.

\begin{proposition}[Proposition 3 from~\cite{raineri2024}]\label{prop:equilibria}
The (uncontrolled) coevolutionary dynamics in \eqref{eq:br} has at least two equilibria: $\vect{x}=\vect{y}=-\vect{1}$ and $\vect{x}=\vect{y}=\vect{1}$, being the unique equilibria in which the action vector is at a consensus ($x_i=x_j$, $\forall\,i,j\in\mathcal V$). 
\end{proposition}

\subsection{Controlled dynamics and problem statement}


We study a scenario in which, at time $t=0$, the population is at one consensus equilibrium and we want to steer it to the opposite one. Without loss of generality, assume the starting consensus is $\vect{x}(0)=\vect{y}(0)=-\vect{1}$, and thus our goal is to reach $\vect{x}=\vect{y}=+\vect{1}$. We consider this problem from the perspective of a policymaker/designer, and assume that our control lever is in the form of directly acting on the state of a subset of agents by setting their opinion and/or action to $+1$ for all $t\geq 1$, yielding the following assumption (illustrated in Fig.~\ref{fig:control}).

\begin{assumption}[Controlled dynamics]
\label{a:initial_condition}
Consider a two-layer network $\mathcal G=(\mathcal V,\mathcal E_A,\mat{A},\mathcal E_W,\mat{W})$ with $\mat{A}$ and $\mat{W}$ stochastic and irreducible. Given ${\mathcal C}^{X}$ the set of controlled actions, and ${\mathcal C}^{Y}$ the set of controlled opinions, there holds
\begin{equation}\label{eq:control_set_def}
    \begin{cases}
        x_i(t) = +1 \; &\forall i \in {\mathcal C}^{X}\;, \forall t \geq 1,\\
        y_j(t) = +1 \; &\forall j \in {\mathcal C}^{Y}\;, \forall t \geq 1,\\
x_i(0) =  -1 \; &\forall i \in \mathcal V \backslash {\mathcal C}^{X}, \\
y_j(0) =  -1 \; &\forall j \in \mathcal V \backslash {\mathcal C}^{Y}.
    \end{cases}
\end{equation}
For $i\notin \mathcal C^{X}$, if $i\in\mathcal R(t)$, then $x_i(t+1)$ follows \eqref{x-dinamic}; for $i\notin \mathcal C^{Y}$, if $i\in\mathcal R(t)$, then $y_i(t+1)$ follows \eqref{y-dinamic}. 
\end{assumption}

\begin{remark}\label{rem:3scenarios}
We identify three scenarios of particular interest:
\begin{enumerate}
    \item {\bf opinion control}, in which one controls only individuals' opinions ($ {\mathcal C}^{X}=\emptyset$);
    \item {\bf action control}, in which one controls only individuas' actions (${\mathcal C}^{Y}=\emptyset$)
    \item {\bf joint control}, in which one controls both variables (${\mathcal C}^{X}={\mathcal C}^{Y}$), which is the special case considered in our preliminary work~\cite{raineri2024}.
\end{enumerate}
These three scenarios reflect different possible real-world interventions, which act only on opinions, on actions, or on both, capturing technical limitations or constraints which may prevent a policymaker from controlling both layers. 
\end{remark}

\begin{figure}
    \centering
\includegraphics{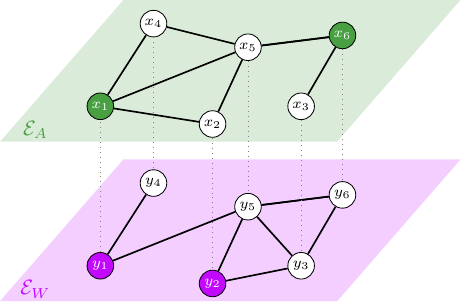}  \caption{Example of a setting that satisfies Assumption~\ref{a:initial_condition}. The top layer represents individuals' action (white for $-1$, green for $+1$), the bottom layer represents individuals' opinion (shades from white for $-1$ to violet for $+1$). Control sets are $\mathcal C^{X}=\{1,6\}$ and $\mathcal C^{Y}=\{1,2\}$.}
    \label{fig:control}
\end{figure}

Hereafter, we will refer to a \emph{controlled coevolutionary dynamics} as a coevolutionary dynamics with utility function in \eqref{eq:utility}, under Assumptions~\ref{a:activation} and~\ref{a:initial_condition}. The goal of the controller, i.e., to lead all the agents to the desired consensus, can be formalized by first defining the objective function 
\begin{equation}\label{eq:objective}
    \phi(\mathcal C^{X},\mathcal C^{Y}):=\mathbb P[\exists\,T<\infty:\vect{x}(t)=\vect{1}, \forall\,t\geq T],
\end{equation}
i.e., the probability (over the probability space generated by the  revision sequence, which might be stochastic, as observed in Remark~\ref{rem:act}) that all individuals definitively switch their action to $+1$ in finite time when the control sets are $(\mathcal C^{X},\mathcal C^{Y})$. The controller's goal is achieved iff $\phi(\mathcal C^{X},\mathcal C^{Y})=1$. Hence, we formalize the following research problem.

\begin{problem}[Effectiveness guarantees]\label{problem}
    Given a network $\mathcal G$, consider a controlled evolutionary dynamics on the network under Assumptions~\ref{a:activation} and~\ref{a:initial_condition} with specified parameters. For given control sets $(\mathcal C^{X},\mathcal C^{Y})$, compute  $\phi(\mathcal C^{X},\mathcal C^{Y})$.
\end{problem}

Solving Problem~\ref{problem} would allow us to determine whether controlling the opinion and/or action of some nodes is sufficient to guarantee convergence to the desired consensus state. At this stage, a second question naturally follows: what is the minimal set of individuals that one should control in order to guarantee that the goal is achieved? Clearly, when controlling an individual, technical limitations may prevent from controlling both actions and opinions, as discussed in Remark~\ref{rem:3scenarios}. Hence, when formulating the problem, we introduce two additional constraints to denote the set of nodes whose action and opinion can be controlled as $\mathcal V^X$ and $\mathcal V^Y$, respectively, allowing us to define the problem as follows. 

\begin{problem}[Minimal control set]\label{problem2}
    Given a network $\mathcal G$, consider a controlled evolutionary dynamics on the network under Assumptions~\ref{a:activation} and~\ref{a:initial_condition} with specified model parameters. Determine the  solution to the following optimization problem
    \begin{equation}\label{eq:minimal}\begin{array}{rl}
      \displaystyle\arg\min\nolimits_{\mathcal C^X\subseteq \mathcal V,\mathcal C^Y\subseteq \mathcal V}&|\mathcal  C^{X}\cup\mathcal  C^{Y}|\\\text{s.t.}&\displaystyle
       \phi(\mathcal C^{X},\mathcal C^{Y})=1,\\&\displaystyle
       \mathcal C^{X}\subseteq \mathcal V^{X},\,\mathcal C^{Y}\subseteq \mathcal V^{Y},\\
       \end{array}
    \end{equation}
    where $\mathcal V^{X}\subseteq \mathcal V$ and $\mathcal V^{Y}\subseteq \mathcal V$ are constraints on the nodes whose action and opinion can be controlled, respectively.
\end{problem}

\begin{remark}
\label{rem:joint control} By setting $\mathcal V^{X}$ and $\mathcal V^{Y}$, one can enforce a specific form for the solution. In fact, by setting $\mathcal V^{X}=\emptyset$ or $\mathcal V^{Y}=\emptyset$, we obtain solutions of Problem~\ref{problem2} with opinion or action control, respectively, i.e., the first two scenarios discussed in Remark~\ref{rem:3scenarios}. On the contrary, if the same constraints are imposed on the two sets ($\mathcal V^{X}=\mathcal V^{Y}$), if a solution to \eqref{eq:minimal} exists, then there is necessarily a solution with joint control ($\mathcal C^X=\mathcal C^Y$), as it will be clear in the next section, after Corollary~\ref{cor:joint}.
\end{remark}


\section{Main properties of the controlled dynamics}\label{sec:results}

In general, the uncontrolled coevolutionary dynamics requires restrictive assumptions to guarantee convergence, such as homogeneous parameters, symmetric layers, and presence of self-loops~\cite{Hassan2023tac} or asynchronous update rules (see Proposition~\ref{prop:super}). For the controlled dynamics, instead, only the mild and general conditions on the activation sequence (Assumption~\ref{a:activation}) and on the stochasticity and irreducibility of the weight matrices (Assumption~\ref{a:initial_condition}) are needed. 

\begin{theorem}
\label{th:convergence}
Consider a controlled coevolutionary dynamics under Assumptions~\ref{a:activation}--\ref{a:initial_condition}. Then, there exists an equilibrium $(\vect{x}^*,\vect{y}^*)$ such that the action vector $\vect x(t)$ converges to $\vect x^*$ in finite time, and the opinion vector $\vect y(t)$ converges to $\vect y^*$ asymptotically. Moreover, both the opinion and action vectors are monotonically nondecreasing functions of time, i.e., $\vect{x}(t+1)\geq\vect{x}(t)$ and $\vect{y}(t+1)\geq\vect{y}(t)$, for all $t\geq 0$.
\end{theorem}

 Since the proof, reported in Appendix~\ref{sec:proof_convergence}, is based on Proposition~\ref{prop:super}, the irreducibility of $\mat W$ and $\mat A$ (equivalently the strong connectivity of each layer of $\mathcal G$) are not strictly required for convergence. However, some of our later results rely on the irreducibility property. If a network is not strongly connected, one can partition it into strongly connected components and control each separately. 

\begin{remark}
Supermodularity of the game implies that all trajectories of the controlled dynamics are lower-bounded by the trajectory with initial condition that satisfies Assumption~\ref{a:initial_condition}. Hence, $\phi(\mathcal C^X,\mathcal C^Y)=1$ is a sufficient condition to reach the desired consensus from any initial condition, not just from the one that satisfies Assumption~\ref{a:initial_condition}, for which the condition is also necessary. Consequently, a solution of Problem~\ref{problem2}, which is optimal (in the sense of the solution being a minimal control set) under Assumption~\ref{a:initial_condition}, will also yield a feasible solution for any initial condition, but optimality  might be lost.
\end{remark}

Theorem~\ref{th:convergence} guarantees that under the  hypotheses of Assumption \ref{a:initial_condition} the controlled coevolutionary dynamics converges and that actions converge in finite time. Moreover, it also guarantees monotonicity of the trajectory of the state vector $\vect z(t)$. As a consequence, if $i\notin\mathcal C^{X}$ switches to action $+1$ at a certain time, then $i$ will never flip back. This observation will be fundamental to development of a systematic approach to addressing our research problems, as presented in Sections~\ref{sec:problem1} and \ref{sec:problem2}. Finally, it is worth noticing that the hypothesis in Assumption~\ref{a:initial_condition} on the uncontrolled node dynamics are key for obtaining monotonicity (which then yields convergence). In fact, from a general initial condition, one may observe non-monotone trajectories (see, e.g.,~\cite{Hassan2023tac}). 
Besides monotonicity, we bring attention to the following property of the controlled dynamics, proved in Appendix~\ref{sec:proof_intermediate}.
\begin{lemma} \label{lemma:intermediate  C}
    Let us consider control sets $(\mathcal C^{X},\mathcal C^{Y})$. If $\phi(\mathcal C^{X},\mathcal C^{Y})=1$, then all control sets $(\mathcal {\bar C}^{\mathcal X},\mathcal {\bar C}^{\mathcal Y})$ such that $\mathcal{C}^{\mathcal X}\subseteq \mathcal {\bar C}^{\mathcal X}$ and $\mathcal {C}^{\mathcal Y}\subseteq \mathcal {\bar C}^{\mathcal Y}$ satisfy $\phi(\mathcal {\bar C}^{\mathcal X},\mathcal {\bar C}^{\mathcal Y}) =1$. If $\phi(\mathcal C^{X},\mathcal C^{Y})=0$, then all control sets $(\mathcal {\bar C}^{\mathcal X},\mathcal {\bar C}^{\mathcal Y})$ such that $\mathcal{\bar C}^{\mathcal X}\subseteq \mathcal {C}^{\mathcal X}$ and $\mathcal {\bar C}^{\mathcal Y}\subseteq \mathcal {C}^{\mathcal Y}$ satisfy $\phi(\mathcal {\bar C}^{\mathcal X},\mathcal {\bar C}^{\mathcal Y}) =0$.
\end{lemma}

An immediate consequence of Lemma~\ref{lemma:intermediate  C} is the following, whose proof is reported in Appendix~\ref{sec:proof_joint}. 
\begin{corollary}\label{cor:joint}If Problem~\ref{problem2} admits a solution, then there is always a solution such that, letting $\mathcal C:=\mathcal C^X\cup \mathcal C^Y$, then  $\mathcal C^X=C\cap\mathcal V^X$ and $\mathcal C^Y=C\cap\mathcal V^Y$.
\end{corollary}
  

In plain words, Corollary~\ref{cor:joint} states that, if it is possible to control both the action and the opinion of a given individual, then it is always optimal to either control both variables or to control none. As a straightforward consequence of Lemma~\ref{lemma:intermediate  C} and Corollary~\ref{cor:joint}, Problem~\ref{problem2} can be reduced to an optimization problem over a single control set, as stated in the following.
\begin{corollary}\label{rem:joint}
  Let
    \begin{equation}\label{eq:minimal2}\begin{array}{rl}
       \mathcal C^*=\displaystyle\arg\min\nolimits_{\mathcal C\subseteq \mathcal V}&|\mathcal  C|\\\text{s.t.}&\displaystyle
       \phi(\mathcal C\cap\mathcal V^{X},\mathcal C\cap\mathcal V^{Y})=1.
       \end{array}
    \end{equation}
    Then, an optimal solution of Problem~\ref{problem2} is given by $C^X=\mathcal C^*\cap\mathcal V^{X}$ and $C^Y=\mathcal C^*\cap\mathcal V^{Y}$.
\end{corollary}

It is worth noticing that, despite the useful properties of the controlled coevolutionary dynamics demonstrated in the above, the problem of controlling the dynamics and determining the  minimal control sets is inherently complex, as stated in the following result, with the proof in Appendix~\ref{sec:proof_np}. 
\begin{theorem}\label{th:np}
Problem \ref{problem2} is NP-complete.
\end{theorem}

Moreover, we can show that the objective function in \eqref{eq:objective} is not submodular, hindering the possibility to easily derive sub-optimal solutions via greedy algorithms~\cite{Nemhauser1978}, as is done for related control problems on social networks~\cite{kempe2003_maxspread}. The proof is reported in  Appendix~\ref{sec:proof_sub}.

\begin{proposition}\label{prop:sub}
The function  $\phi(\mathcal C^{X},\mathcal C^{Y})$ in \eqref{eq:objective} is not submodular with respect to any of its two variables.
\end{proposition}



\section{Effectiveness Guarantees Problem} \label{sec:problem1}
 The results from the previous section call for the development of algorithms able to solve our research problem in an efficient way. We start from Problem \ref{problem}. In our preliminary work \cite{raineri2024}, we proposed an algorithm and conjectured that this algorithm can solve Problem~\ref{problem} in the simplified setting of joint control, $\mathcal C^{X}=\mathcal C^{Y}$ (see Remark~\ref{rem:3scenarios}). Here, we propose a refined version of the algorithm that accounts for the more general setting in Assumption \ref{a:initial_condition}, while also  reducing the total computational complexity. Then, we rigorously prove the conjecture in the general scenario, demonstrating that the proposed algorithm solves Problem~\ref{problem} in polynomial time. Our procedure is based on the following iterative scheme.

We consider the general case in which we control the action for nodes in $\mathcal{C}^{X}$ and the opinion for nodes in $\mathcal{C}^{Y}$ as described in Assumption \ref{a:initial_condition}. At iteration $k=1$, we initialize the algorithm by defining $\mathcal A(1)= \mathcal{C}^{X}$, which involve only the nodes whose action is controlled. At each step of the algorithm $k$, we construct a candidate equilibrium with action vector $\vect{\hat x}$ with
\begin{equation}\label{eq:x-candidate}
    \hat x_i= \left\{\begin{array}{ll}
        +1&\text{if }i\in\mathcal A(k),\\
        -1&\text{if }i\notin\mathcal A(k),
        \end{array}\right.
\end{equation}
and opinion vector $\vect{\hat y}$, computed by solving the linear system
\begin{equation}
\label{y-equilibrium}
\hat y_i=\left\{\begin{array}{ll}(1-\lambda_i)\sum_{j \in \mathcal V} w_{ij}\hat y_j + \lambda_i  \hat x_i&\text{if }i\notin \mathcal{C}^{Y},\\
+1&\text{if }i\in \mathcal{C}^{Y},
\end{array}\right.
\end{equation} 
which has a unique solution, as we will prove later. Observe that solving \eqref{y-equilibrium} requires inverting a matrix that is independent of the variables. This operation can be optimized by computing it in advance, before running the iterations.

Then, we will demonstrate in  Theorem~\ref{th:algorithm} below the following properties. First, we show that given an action vector $\vect{\hat x}$, there exists a unique $\vect{\hat y}$ such that $\vect{\hat z}=(\vect{\hat x},\vect{\hat y})$ is a candidate equilibrium of the controlled coevolutionary dynamics. 
To check whether $\vect{\hat z}$ is an actual equilibrium, we check if any individual $i$ who plays action $-1$ at $\vect{\hat z}$ would switch to $+1$. According to Proposition~\ref{prop:dynamics}, this can be checked by computing the sign of $\delta_i(\vect{\hat z})$ for all $i\notin\mathcal A(k)$. If all $\delta_i(\vect{\hat z})\leq 0$, then no individual will switch action, and the candidate $\vect{\hat z}$ is indeed the equilibrium reached by the system. 
Otherwise, we will prove that all individuals with $\delta_i(\vect{\hat z})>0$ will eventually switch to $+1$. Hence, $\vect{\hat z}$ is not an equilibrium of the controlled coevolutionary dynamics, and we need to consider other potential equilibria where also those individuals with $\delta_i(\vect{\hat z})>0$ switch to action $+1$. To this aim, we increase the iteration index $k$ by $1$, and we enlarge the set $\mathcal A(k)$ by incorporating these individuals into $\mathcal A(k-1)$, and we iterate the procedure, 
until the termination criterion $\mathcal A(k) = \mathcal A(k-1)$, which implies that no more individuals would change action. According to this procedure, we get a non-decreasing sequence of sets. When the termination criterion  is met, the algorithm returns $\mathcal A_f$. 

This algorithm, for which a computationally-improved pseudo-code is reported in Algorithm~\ref{alg}, offers a tool to solve Problem~\ref{problem} in a polynomial time, as summarized in the following statement, whose proof is reported in Appendix~\ref{sec:proof2}.
\begin{algorithm}
\caption{Equilibrium computation
\label{alg}
}
\KwData{$\mat{A},\mat{W},\mathcal C^{X}, \mathcal C^{Y}, \lambda_i$ and $\beta_i,$ for all $i\in\mathcal U$}
\KwResult{$\mathcal A_f:=\mathcal A(k)$, i.e., individuals with $x^*=+1$}
$k \gets 1;\; \mathcal A(0) \gets \emptyset;\; \mathcal A(1) \gets \mathcal C^{X}$; ${\hat y}_i \gets +1$ $\forall i \in \mathcal{C}^{Y}$\;
$\mat{M}\gets(\mat{I}-(\mat I- \text{diag}(\vect\lambda))\mat W)^{-1}$\;
\While{$\mathcal A(k) \neq \mathcal A(k-1)$}{
Define $\vect{\hat x}$ using \eqref{eq:x-candidate};

${\hat y}_i\gets (\mat M \text{diag}(\vect \lambda) \vect{\hat x})_i$ for all $i \notin \mathcal{C}^{Y}$ ;

$k \gets k+1;\;
\mathcal A(k) \gets \mathcal A(k-1)$\;
check    \For{$i \in \mathcal V \And i \notin\mathcal A(k)$}
    {
        \If{$\delta_i(\vect{\hat x},\vect{\hat y})>0$} 
        {$\mathcal A(k) \gets \mathcal A(k) \cup \{i\}$;}
    }
}
\end{algorithm}

\begin{theorem}\label{th:algorithm}
     Under Assumptions~\ref{a:activation} and~\ref{a:initial_condition}, Algorithm~\ref{alg} solves Problem~\ref{problem} in time $O(n^3)$. In fact, given control sets $(\mathcal C^{X},\mathcal C^{Y})$ and output $\mathcal A_f$ of Algorithm~\ref{alg}, then
$$
        \phi(\mathcal C^{X},\mathcal C^{Y})=\left\{\begin{array}{ll}1&\text{if }\mathcal A_f=\mathcal V,\\
        0&\text{otherwise.}\end{array}\right.
$$
    Moreover, the equilibrium reached by a controlled coevolutionary dynamics that satisfies Assumptions~\ref{a:activation} and~\ref{a:initial_condition} with control sets $(\mathcal C^{X},\mathcal C^{Y})$ is $(\vect{x}^*,\vect y^{*})$, with $\vect{x^*}$ defined as in \eqref{eq:x-candidate} with $\mathcal A(k)=\mathcal A_f$ and $\vect{y}^*$ the solution of \eqref{y-equilibrium} given $\vect{x}^*$.
\end{theorem}

In most practical scenarios, only a few iterations are needed for convergence, since $\mathcal A(k)$ often increases by more than one individual at each iteration, further reducing the computational effort needed. Moreover, the heaviest operation is the computation of matrix $\mat M=(\mat I - [\vect 1-\diag(\vect \lambda)]\mat W)^{-1}$, which is used to solve \eqref{y-equilibrium}. However, $\mat M$ is independent of the control sets. Hence, it can be precomputed and used for multiple instances of Algorithm~\ref{alg}. Finally, one can leverage symmetry properties of the network structure to reduce the dimension of the system, as illustrated in the following example.

\subsection{Complete Graph}\label{sec:complete}
We illustrate a motivational example showing how our results can be used to analyze the coevolutionary dynamics and the possibility to control it. Specifically, we consider the case of a complete graph (including self-loops) with homogeneous parameters and weights, in which, due to symmetry reasons, the strategy adopted guarantees a solution for both Problems~\ref{problem} and~\ref{problem2}. 
 In order to check if a candidate control set succeeds in complete network controllability we use Algorithm \ref{alg}.
\begin{assumption}[Homogeneous complete graph]
\label{a:homogeneous}
Let $\mathcal G$ be a two-layer network with $a_{ij}=w_{ij}= \frac{1}{n-1}$ and $a_{ii}=w_{ii}=0$, $\forall\,i\neq j \in \mathcal{V}$. Moreover, let $\lambda_i=\lambda$ and $\beta_i=\beta$,  $\forall\,i\in\mathcal V$.
\end{assumption}

In \cite{raineri2024}, we have analysed the scenario of a complete graph with joint control (see Remark~\ref{rem:3scenarios}).  Here, we focus on the other two scenarios of interest discussed in Remark~\ref{rem:3scenarios}. More precisely, we assume that we are able to control a certain number of individuals (whose position is irrelevant, due to the network symmetry) and  we also consider the two scenarios of opinion and action control. Using Algorithm~\ref{alg} and Theorem~\ref{th:algorithm}, we  establish the following result. 

	\begin{proposition}\label{prop:complete}
Consider a coevolutionary dynamics that satisfies Assumptions \ref{a:activation}--\ref{a:initial_condition} on a complete graph $\mathcal G$ with $n$ nodes that satisfies Assumption \ref{a:homogeneous}.   Given $\mathcal C\subseteq \mathcal V$, let  $\gamma=\frac{|\mathcal C|}{n-1}$. Then, 
\begin{itemize}
           \item[i)] for opinion control, $\phi(\emptyset,\mathcal C)=1$  iff it holds
    \begin{equation}\label{eqn:complete_controlY}
           \frac{\beta[3(1-\lambda)\gamma +2\lambda\gamma (1-\lambda)+\lambda(2\lambda-1)]}{\gamma +\lambda-\lambda\gamma}>1;
    \end{equation}
    \item[ii)] for action control, $\phi(\mathcal C,\emptyset)=1$ iff $\gamma > 1/2$;
\item[iii)] for joint control, $\phi(\mathcal C,\mathcal C)=1$ iff it holds
\begin{equation}\label{eqn:complete_controlXY}
    2 \beta (1-\lambda)\Big(\frac{   \gamma-\lambda+ \lambda\gamma}{\gamma +\lambda-\lambda\gamma}\Big)+(1-\beta)(2\gamma-1)>0.
\end{equation} 
    \end{itemize}
	\end{proposition}
\begin{proof}
   In all cases, we apply Algorithm \ref{alg} and Theorem~\ref{th:algorithm}.

    i) We start with $\mathcal A(1)=\emptyset$.
    The candidate equilibrium $\vect{\hat z}$ has action vector $\vect{\hat x}=-\vect{1}$ and opinion vector $\vect{\hat y}$
    with $\hat{y}_i=+1$ for all $i\in\mathcal{C}$, and by solving  \eqref{y-equilibrium}  for all $i\notin\mathcal{C}$. By symmetry, all  $i\notin\mathcal{C}$ have necessarily $\hat{y}_i=+1$, solution of $$\hat{y}_i=(1-\lambda)(\gamma+(1-\gamma)\hat y_i)-\lambda,$$ 
    which yields $$\hat y_i=\frac{(1-\lambda)\gamma -\lambda}{1-(1-\lambda)(1-\gamma)}.$$ Substituting this in  $\delta_i(\vect{\hat z})$, we get that $\delta_i(\vect{\hat z})>0$ iff the condition in \eqref{eqn:complete_controlY} is satisfied. In this scenario, $A(2)=\mathcal V$; otherwise, $A(2)=A(1)$. In both cases the algorithm terminates, and Theorem~\ref{th:algorithm} yields  claim i).
   
   ii) We start with $\mathcal A(1)=\mathcal C$. The corresponding candidate equilibrium $\vect{\hat z}$ has $\vect{\hat x}$ defined using \eqref{eq:x-candidate}. Conversely, for the opinion vector $\vect{\hat y}$ we have to distinguish two different cases: $\hat y_C$ indicates the equilibrium value for those nodes whose action is controlled, $\hat y_U$ the one for those nodes which are not controlled at all. The equilibrium vector is so defined as follows:
   \begin{equation}\label{eq:systemC}
   \begin{cases}
       \hat y_C &= (1-\lambda)[\gamma \hat y_C+(1-\gamma)\hat y_U]+\lambda\\
       \hat y_U &= (1-\lambda)[\gamma \hat y_C +(1-\gamma)\hat y_U]-\lambda.
   \end{cases}
   \end{equation}
   Solving \eqref{eq:systemC} and substituting the solution in  \eqref{eq:delta} for any $i\notin\mathcal C$, we obtain 
   $\delta_i(\vect{\hat z})=(2\gamma -1)((1-2\lambda)\beta \lambda+1)$. It is easy to verify that $((1-2\lambda)\beta \lambda+1)>0$ for any choice of  $\lambda$ and $\beta$, yielding the condition $\gamma>1/2$. 
   If $\gamma>1/2$, then $\delta_i(\vect{\hat z})>0$, and $\mathcal A(2)=\mathcal V$; otherwise,  $\mathcal A(2)=\mathcal A(1)=\mathcal C$. In both cases, the algorithm terminates, and Theorem~\ref{th:algorithm} yields ii).

   iii) We apply Theorem~\ref{th:algorithm} to \cite[Proposition~3]{raineri2024}.
\end{proof}

\begin{remark}
Proposition~\ref{prop:complete} solves both Problems~\ref{problem} and~\ref{problem2} for a complete graph. In fact, given control sets $(\mathcal C^{X},\mathcal C^{Y})$, where  $\gamma=\frac{|\mathcal C^{X}\cup\mathcal C^{Y}|}{n-1}$, Problem~\ref{problem} is solved for those values of $\lambda$ and $\beta$ that satisfy the condition corresponding to the scenarios considered from Remark~\ref{rem:3scenarios}. On the other hand, given $\gamma$, the choice of the nodes to control is irrelevant, and re-writing \eqref{eqn:complete_controlY} and \eqref{eqn:complete_controlXY} as  conditions on $\gamma$, we can ultimately determine the minimum number of nodes to be controlled  to guarantee $\phi(\mathcal C^{X},\mathcal C^{Y})=1$, solving Problem~\ref{problem2}.
\end{remark}

From Proposition~\ref{prop:complete} we can draw some interesting conclusions. First, from item ii), we observe that for action control, the problem is exactly the same as controlling a majority on a complete network: we need more than 50\% of nodes being controlled. On the contrary, the conditions become  nontrivial when considering opinion control and joint control, since they depend on the model parameters. Predicatively, we can show that \eqref{eqn:complete_controlY} and $\gamma>1/2$ are always more restrictive than \eqref{eqn:complete_controlXY}, implying that joint control (if possible) is always preferable. Then, further observations can be made by plotting the minimal value of $\gamma$ that satisfies \eqref{eqn:complete_controlY} and \eqref{eqn:complete_controlXY}. 

These plots, reported in Fig.~\ref{fig:complete} using color-code intensity,  illustrate how the model parameters strongly impact the cardinality of the controlled set $\mathcal{C}$, especially for opinion control. In fact, we observe that for small values of $\beta$ and large $\lambda$, effective policies should entail opinion control of a large majority of the population. A similar dependency is also present for joint control, but it is less pronounced. Interestingly, when $\beta$ is large and $\lambda$ is small-to-moderate (top left corner in Fig.~\ref{fig:complete_opinion_control}), we observe that opinion control is more effective than action control. In fact, the required fraction of population to control can be less than $50\%$, which is instead the case for action control. In this vein, Proposition~\ref{prop:complete} can be used to design optimal policies to facilitate the network controllability or increase its robustness against malicious attacks. 

\begin{figure}
 \centering
\subfloat[Opinion Control]{
    \includegraphics[width=.46\linewidth]{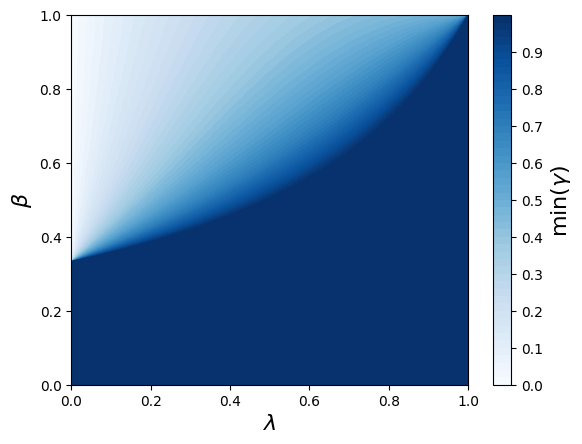}  \label{fig:complete_opinion_control}}\subfloat[Joint Control]{  \includegraphics[width=.46\linewidth]{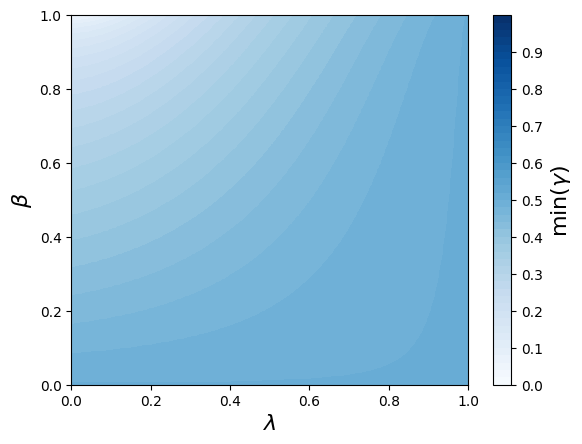}  \label{fig:complete_joint_control}}

   \caption{Results for a complete graph.  The color intensity represents the cardinality of the minimal control set $\gamma=|\mathcal C|/(n-1)$ that solves Problem~\ref{problem2}. }
    \label{fig:complete}
\end{figure}

\section{Minimal Control Set Identification} \label{sec:problem2}

In the previous section, we have shown how Algorithm~\ref{alg} can be used to solve our research problems for simple network structures such as a complete graph (Section~\ref{sec:complete}) or a star graph (see~\cite{raineri2024}). However, in more realistic scenarios, the network does not have such a level of symmetry, and Algorithm \ref{alg} is not efficient to solve Problem \ref{problem2}. In fact, in order to find the optimal control sets $\mathcal C^X$ and $\mathcal C^Y$, one should look for all pairs of subsets of $\mathcal V$ to find those that guarantee that the constraint $\phi(\mathcal C^X,\mathcal C^Y)=1$ is satisfied. First, we observe that Corollary~\ref{rem:joint} allows us to simplify the problem (and the notation), since we simply need to find an optimal control set $\mathcal C$, and then define the solution of Problem~\ref{problem2} as $(\mathcal C^X, \mathcal{C}^Y)$ with 
\begin{equation}\label{eqn:C_definition}
\mathcal C^X=\mathcal C\cap \mathcal V^X \; \text{and} \; \mathcal C^Y=\mathcal C\cap \mathcal V^Y\,.\end{equation} 
However, even with this simplification, the determination of the minimal control set remains computationally challenging. 

With an aim to reduce the computational cost to perform such a task, we now build on the optimal targeting algorithm proposed in \cite{Como2022supermodular} in order to design an algorithm able to solve Problem~\ref{problem2}. Intuitively, the rationale behind the proposed methodology consists of starting from controlling all the possible individuals and moving backwards, removing individuals from the control set, until we find the minimal control set that allows us to reach the goal in Problem~\ref{problem2}. However, such a naive approach might result in being stuck in a local minimum, not being able to reach the global optimum. For this reason, we employ a stochastic approach, which consists of defining a discrete-time Markov chain~\cite{levin2006book} that explores the space of control sets that are feasible solutions of Problem~\ref{problem2} in such a way that its invariant distribution will provably concentrate about the global optimal solutions of \eqref{eq:minimal2}, allowing us to define an effective heuristics to solve Problem~\ref{problem2}.

More precisely, given Problem~\ref{problem2} with constraints $\mathcal C^X\subseteq \mathcal V^X$ and $\mathcal C^Y\subseteq \mathcal V^Y$, we define the set of all \emph{controllable nodes} as
$$    \mathcal V^*:=\mathcal V^X\cup \mathcal V^Y,
$$
and we let $n^*:=|\mathcal V^*|$ to be the number of controllable nodes. Then, we can define the space of all \emph{potential control sets}, which is nothing but the power set of $\mathcal V^*$, i.e., $\mathscr C:=\{\mathcal C:\mathcal C\subseteq \mathcal V^*\}$. 
Moreover, we say that a potential control set $\mathcal C\in \mathscr{C}$ is \emph{admissible} if and only if $\phi(\mathcal C\cap \mathcal V^X,\mathcal C\cap \mathcal V^Y)=1$. In other words, an admissible control set is a feasible solution of \eqref{eq:minimal2}. We indicate the space of all admissible control sets as
$$
    \bar{\mathscr C}:=\{\mathcal C\subseteq \mathscr{C}:\phi(\mathcal C\cap \mathcal V^X,\mathcal C\cap \mathcal V^Y)=1\},
$$
which is clearly the set of all feasible solutions of Problem~\ref{problem2}.


Our algorithm starts from the worst case in which we control all the controllable nodes, i.e., we set $\mathcal C = \mathcal V^*$. First, we need to check whether $\mathcal C$ is admissible, i.e., we define $\mathcal{C}^X$ and $\mathcal{C}^Y$ as in \eqref{eqn:C_definition}, and check whether $\phi(\mathcal{C}^X, \mathcal{C}^Y)=1$. This check is done by employing Algorithm~\ref{alg}. If $\mathcal C$ is admissible, then it follows that $(\mathcal C^X,\mathcal C^Y)$ is a feasible solution   to Problem~\ref{problem2}. Conversely, if $\mathcal C$ is not admissible, then Problem~\ref{problem2} is unfeasible and there is no need to proceed. Trivially, we observe that if $\mathcal{V}^X=\mathcal{V}$ then $\mathcal C$ is always admissible, being $\mathcal C^X=\mathcal V$, and $\phi(\mathcal V,\mathcal C^Y)=1$ for any choice of $\mathcal C^Y$.

If Problem~\ref{problem2} is feasible, we then adopt the following iterative procedure, which is detailed in Algorithm \ref{alg:optimal_C}. At the $k$th iteration we start with the control set $\mathcal C$. We select a node $r$, uniformly at random among the controllable nodes, i.e., $r\in\mathcal V^*$. Then, two cases are possible:
    \begin{enumerate}
        \item The node belongs to the control set ($r\in\mathcal C$). In this case, if the set $\mathcal C\setminus\{r\}$ is admissible (which is checked using Algorithm~\ref{alg}), then the control set is updated to $\mathcal C\setminus\{r\}$;  otherwise it remains $\mathcal C$; or
        \item The node does not belong to the control set ($r\notin\mathcal  C$). In this case, we introduce a probability $\varepsilon\in(0,1]$, and  the node $r$ is added to the control set with probability $\varepsilon$, i.e, the control set is updated to $\mathcal C\cup\{r\}$;  otherwise, $\mathcal C$ remains unchanged.
    \end{enumerate}
The iteration counter is thus increased to $k+1$ and the process is repeated. As a design choice, we will consider a maximum number of iterations for the algorithm equal to $T$.

\begin{algorithm}
\caption{Optimal control set identification \label{alg:optimal_C}
}
\KwData{$\mat{A},\mat{W}, \lambda$, $\beta$, $\varepsilon$,  $n$, $T$, $\mathcal{V}^X$, and $\mathcal{V}^Y$ }
\KwResult{$\hat{\mathcal{C}}$ such that $(\hat{\mathcal C}^X, \hat{\mathcal{C}}^Y)$ solves Problem~\ref{problem2}  }
$k \gets 1;$ 
$\mathcal{C} \gets \mathcal{V}^X \cup \mathcal{V}^Y$; $\hat{\mathcal{C}}\gets \emptyset$\;
{\If{$\phi(\mathcal{C}^X, \mathcal{C}^Y)=1$}{
$\hat{\mathcal{C}}\gets \mathcal{C}$\;
{\While{$k < T$}{
$k \gets k+1$;
Choose at random a node $r \in \mathcal V^X$\;
\If{$r \in \mathcal C$}
{\If{ Algorithm \ref{alg} yields $\mathcal{A}_f = \mathcal{V}$}
{$\mathcal C \gets \mathcal C \backslash \{r\}$ \;
{\If{$|\mathcal{C}|<|\hat{\mathcal{C}}|$}
{$\hat{\mathcal{C}}\gets \mathcal{C}$\;
}}
}
}
\Else{
$\mathcal C\gets \mathcal C \cup \{r\}$ with probability $\varepsilon$ 
}
 }
 }
 }
 }
\end{algorithm}

Before formally presenting the Markov chain induced by this iterative process and illustrating how this can be used to solve Problem~\ref{problem2}, we offer here a simple example to elucidate the procedure described above. We consider a network with $n=4$ nodes and with $\mathcal V^*=\mathcal V$. We start, at an arbitrary iteration step, from a control set $\mathcal C=\{1,2,3\}$. Figure~\ref{fig:config_space} illustrates the three elements of $\mathscr{C}$ that can be reached from $\mathcal C$, depending on which node $r$ is selected. Assume that Algorithm~\ref{alg} prescribes that only the two sets on the left belong to $\mathscr{\bar C}$ and are thus admissible, while the third one is not admissible, and it is thus barred in Fig.~\ref{fig:config_space}. If 
nodes $r=2$ or $r=3$ are selected, then we are in step 1) 
and the chain has a transition to $\mathcal{C}\setminus\{2\}$ or $\mathcal{C}\setminus\{3\}$, respectively. If instead node $r=1$ is selected, we are in step 1) 
but the chain remains in $\mathcal{C}$, since $\mathcal{C}\setminus\{1\}$ is not admissible. Finally, if node $r=4$ is selected, we are in step 2) 
and the chain transitions to $\mathcal{C}\cup\{4\}$ with probability $\varepsilon$, as illustrated in Fig.~\ref{fig:single flip}.


   \begin{figure}
       \centering
       \subfloat[Admissible and not admissible sets]{\includegraphics{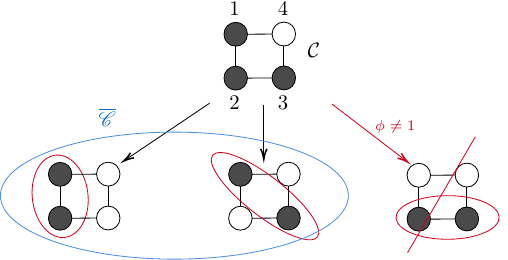}\label{fig:config_space}}\\
       \subfloat[Transitions of the chain]{\includegraphics[width=.98\linewidth]{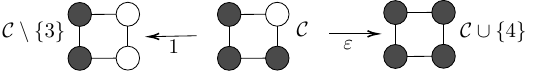}   \label{fig:single flip}}
       \caption{Example of one iteration of the Markov chain from set $\mathcal{\bar C}$. In (a), we highlight with a blue circle the  set of admissible control sets.  The red circles highlight all the candidate control sets. The last control set is not admissible and so it is not considered. In (b), we illustrate the transitions of the chain, described in \eqref{eq:MC_probabilities}. Nodes in the control set are denoted in black, nodes not in the control set in white. }
       \label{fig:config_space_all}
   \end{figure}

\begin{proposition}\label{prop:MC}
The procedure described in Algorithm~\ref{alg:optimal_C} induces a   discrete-time Markov chain $Z_\varepsilon(t)$, which is defined on the space of admissible control sets $\bar{\mathscr C}$, has initial state $Z_\varepsilon(t)=\mathcal V^*$, and, given any pair $\mathcal A,\mathcal B\in \bar{\mathscr C}$, its transition probabilities are defined as
$$
    \mathcal P[Z_\varepsilon(t+1)=\mathcal B|Z_\varepsilon(t)=\mathcal A]=P_{\mathcal A, \mathcal B,\varepsilon},
$$
where
\begin{equation}\label{eq:MC_probabilities}
P_{\mathcal A, \mathcal B,\varepsilon} = \begin{cases}
    1/n^* \quad& \text{if}\;\mathcal B \subset \mathcal A \; \text{and}\; |\mathcal B|= |\mathcal A|-1, \\
    \varepsilon / n^* \quad& \text{if}\; \mathcal B \supset \mathcal A \;\text{and}\; |\mathcal B|= |\mathcal A|+1, \\
    1-\alpha_\varepsilon(\mathcal A) \quad& \text{if}\; \mathcal B = \mathcal A, \\
    0 \quad& \text{otherwise},
\end{cases}
\end{equation}
with $\alpha_\varepsilon(\mathcal A) = \frac{\varepsilon(n^*-|\mathcal A|)+n_c(\mathcal A)}{n^*}$, 
where $n_c(\mathcal A)$ is the number of admissible configuration that can be reached by removing a node from $\mathcal A$, i.e., $n_c(\mathcal A):=|\{\mathcal C\in\mathscr{\bar C}:\mathcal C=\mathcal A\setminus\{r\},r\in\mathcal A\}|$. 
\end{proposition}
\begin{proof}
    First, we observe that the iterative procedure described in Algorithm~\ref{alg:optimal_C} explores only admissible control set. We proceed by induction. If Problem~\ref{problem2} is feasible, then we start from $\mathcal C=\mathcal V^*$, which is feasible. Then, if at the $k$th iteration the set  $\mathcal C$ is an admissible control set,  we demonstrate that this holds true also at iteration $k+1$. In fact, if 1) 
    occurs, then either $\mathcal C\setminus\{r\}$ is  admissible by construction or the control set remains unchanged; if 2) 
    occurs, then the control set is updated to a superset  of $\mathcal C$ (possibly coinciding with $\mathcal C$ with probability $1-\varepsilon$), which is  admissible due to Lemma~\ref{lemma:intermediate  C}, yielding that Algorithm~\ref{alg:optimal_C} induces a stochastic process with state space $\mathscr{\bar C}$. 
    
    Second, let us denote by
    $\mathcal A$ and $\mathcal B$ the control set a the $k$th and $(k+1)$th iteration of Algorithm~\ref{alg:optimal_C}, respectively. Preliminary, we observe that, given $\mathcal A$, $\mathcal B$ is is independent of the previous history of the process; hence, it is a Markov chain~\cite{levin2006book}. Then, if node $r$ (selected uniformly at random among $n^*$ nodes) is such that $r\in\mathcal A$, then $\mathcal B=\mathcal A\setminus\{r\}$. Hence, a generic set $\mathcal B\subset A$ with $|\mathcal B|=|\mathcal A|-1$ is  reached with probability $1/n^*$, yielding the first line in \eqref{eq:MC_probabilities}. If the node $r\notin\mathcal A$, then $\mathcal B=\mathcal A\cup\{r\}$ with probability $\varepsilon$. Hence, a generic set $\mathcal B\supset A$ with $|\mathcal B|=|\mathcal A|+1$ is  reached with probability $\varepsilon/n^*$, yielding the second line in \eqref{eq:MC_probabilities}. No other state can be reached according to the algorithm.   
    Therefore, we conclude that these probabilities match exactly those in \eqref{eq:MC_probabilities}, where the third line is simply obtained as the probability of the complementary event, yielding the claim.
\end{proof}

The Markov chain defined in Proposition~\ref{prop:MC} plays a key role in solving Problem \ref{problem2}, as claimed in the following statement. 
\begin{theorem}\label{th:convergence_MC}
    Algorithm~\ref{alg:optimal_C} induces a Markov chain whose invariant distribution $\mu_\varepsilon\in[0,1]^{\mathscr{\bar C}}$ is such that $\lim_{\varepsilon \searrow 0} \mu_\varepsilon = \mu$ where $\mu$ is the uniform probability distribution on the set of solutions of \eqref{eq:minimal2} and, consequently, of Problem~\ref{problem2}.
\end{theorem}
\begin{proof}
We observe that the Markov chain $Z_\varepsilon(t)$ with transition probabilities in  \eqref{eq:MC_probabilities} is ergodic, since every admissible configuration $\mathcal B\in\mathscr{\bar C}$ can be reached from any other one $\mathcal A\in\mathscr{\bar C}$ following a path of non-zero probability (trivially, first by adding nodes to  $\mathcal A$ until reaching $\mathcal V^*$, and then by removing nodes until reaching $\mathcal B$)~\cite{levin2006book}. This path passes only through admissible control sets, due to Lemma~\ref{lemma:intermediate  C}. Hence, the chain converges to an invariant distribution. 

To compute its invariant distribution, we follow the arguments from  \cite[Theorem 2]{Como2022supermodular}, to conclude that $\mathcal Z_\varepsilon(t)$ has invariant distribution $\mu_\varepsilon\in[0,1]^{\mathscr{\bar C}}$, such that its generic component associated with $\mathcal C\in\mathscr{\bar C}$ is equal to $[\mu_\varepsilon]_{\mathcal C}=\frac{1}{K_\varepsilon}\varepsilon^{|\mathcal C|}$, where $K_\varepsilon$ is a normalizing coefficient. Hence, for $\varepsilon\searrow 0$, it holds that the only non-zero components of $\mu_\varepsilon$ are all equal and are those associated with the admissible control sets $\mathcal C$ of minimal cardinality, yielding the claim. 
\end{proof}

\begin{remark}
    The convergence result proved in Theorem~\ref{th:convergence_MC} guarantees that, in the limit $t\to\infty$, the Markov chain defined in Proposition~\ref{prop:MC} concentrates about the solution(s) of Problem~\ref{problem2}. From a practical point of view, in Algorithm~\ref{alg:optimal_C} we incorporate a variable $\mathcal{\hat C}$ to keep track of the minimal control set reached so far, so that it could be used also as an heuristics for admissible configurations space exploration. 
\end{remark}
\begin{remark}
The computations required to solve Algorithm \ref{alg:optimal_C} may be  reduced. First, the most computationally complex operation is  associated with the solution of \eqref{y-equilibrium} in Algorithm~\ref{alg}. However, this operation, which is performed by using matrix $\mat M$, is independent of the iteration. Hence, Algorithm~\ref{alg:optimal_C} can be optimized by computing the matrix $\mat M$ before starting the iterations and providing it as an input to Algorithm~\ref{alg}, so that the inversion of a matrix needs to be performed only once. Second,  if $\mathcal{V}^Y=\emptyset$, i.e., we enforce action control, then in order to verify that a control set obtained by removing a node $r$ from an admissible control set is admissible, we need not apply Algorithm~\ref{alg} thoroughly. In fact, it is enough to stop the iteration in Algorithm~\ref{alg} as soon as $\delta_r(\vect{\hat x},\vect{\hat y})>0$, due to Lemma~\ref{lemma:intermediate  C} and the monotonicity property of actions $\vect x$. 
\end{remark}

\section{Case Study}\label{sec:case}

\begin{figure}
    \centering
    \includegraphics[width=\linewidth]{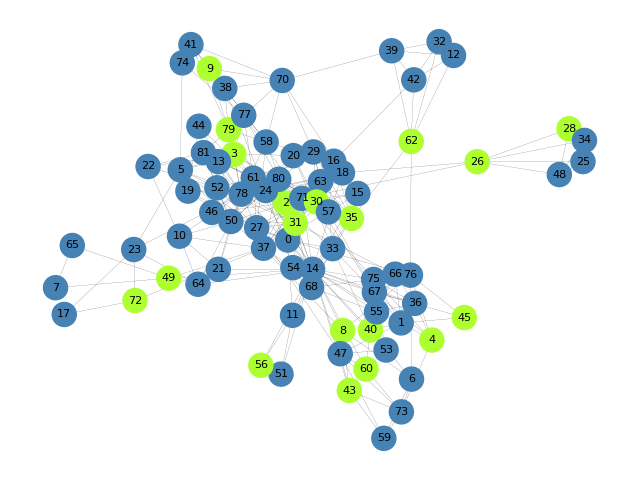}\vspace{-.3cm}
    \caption{The network used in  Section~\ref{sec:case}. Green nodes are those identified by our algorithm as control nodes, in the scenario of joint control.  }
    \label{fig:malawi}
\end{figure}

We demonstrate our approach on a real-world case study of a network associated with social contacts in a village in rural Malawi, whose dataset is available on Sociopatterns~\cite{malawi_net}. We use the largest connected component of the network, consisting of 84 individuals and 346 weighted undirected edges, illustrated in Fig.~\ref{fig:malawi}. 
We apply Algorithm~\ref{alg:optimal_C} with the exploration parameter $\varepsilon$ set to be $\varepsilon=0.1$ and $\varepsilon=0.6$, and we keep track of how the quality of the optimal solution (in terms of fraction of nodes to be controlled) evolves as the number of iterations increases. The results are reported in Fig.~\ref{fig:malawi_comparison}. 

\begin{figure}
 \centering
\subfloat[$\varepsilon = 0.1 $]{
    \includegraphics[width=0.48\linewidth]{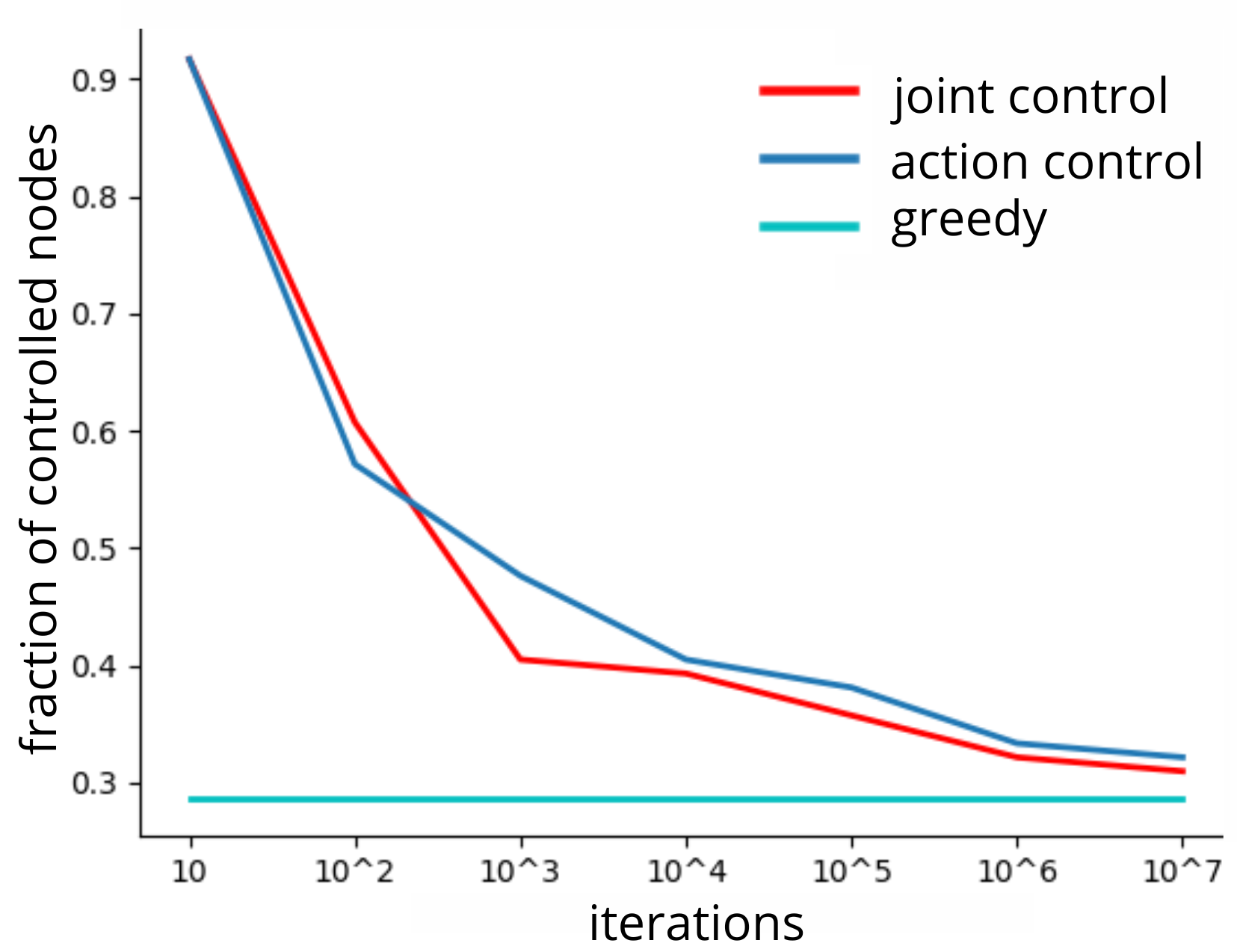}   \label{fig:Malawi01_comparison}}
\subfloat[$\varepsilon = 0.6 $]{\includegraphics[width=0.48\linewidth]{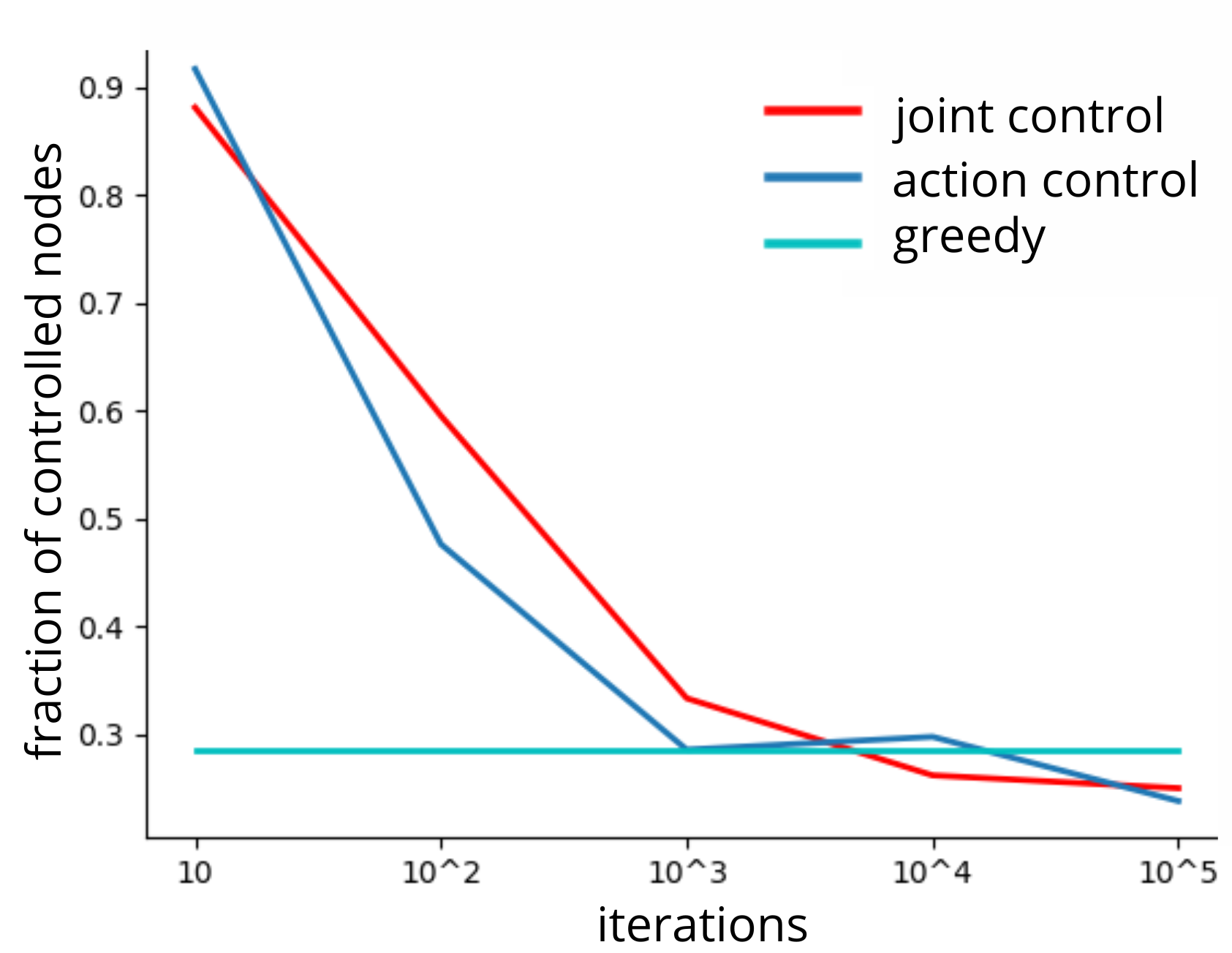}  \label{fig:Malawi01_comparison_esp06}}
   \caption{Comparison of control strategies for different values of $\varepsilon$. Blue and red curves represent the fraction of agents to be controlled with joint and action control, respectively. In cyan, the results obtained with a greedy heuristics. }
    \label{fig:malawi_comparison}
\end{figure}

First, we observe that increasing the number of iterations of the algorithm significantly reduces the fraction of nodes that need to be controlled. As expected, in the long run, joint control outperforms action control, i.e., controlling only the $x$ values. However, since the performance difference is not substantial, controlling only the agents' actions, which one expects would require a weaker enforcement policy, could be a valid choice to reduce costs. On the contrary, opinion control is not sufficient to steer the system to the desired consensus.

Second, a remarkable result is highlighted in Fig.~\ref{fig:malawi_comparison}. Specifically, for a sufficiently large value of  $\varepsilon$ and number of algorithm iterations,  Algorithm~\ref{alg:optimal_C} outperforms a greedy heuristics algorithm that selects nodes with larger Bonacich centrality as control nodes. For instance, given $\varepsilon=0.6$, after $100,000$ iterations (performed in only 70 seconds on a standard PC) the minimum control set cardinality required is reduced to $23.8\%$, outperforming the $26.2\%$ obtained with the greedy algorithm (cyan line in Fig.~\ref{fig:Malawi01_comparison_esp06}). This underlines not only the effectiveness of our algorithm, but also the importance of fine-tuning of its parameters based on the specific case study at hand, as the optimal parameters may vary significantly depending on the context. Finally, it is worth noticing that the fraction of nodes needed to be controlled in this real-world case study is comparable with the critical mass of innovators observed in different experimental studies~\cite{centola2018experimental_tipping,ye2021nat}, providing further empirical support to the coevolutionary model studied in this paper.


\section{Conclusion}\label{sec:conclusion}

In this paper, we formalized a novel control problem for social networks by incorporating a committed minority in a coevolutionary model of actions and opinions. By analyzing the controlled model, we established a general convergence result and we leverage it to tackle two research problems: i) determine whether the committed minority are able to steer the  population to the desired state and ii) identify the minimal control set needed to achieve the goal. By developing algorithms to address these two questions, we offer a novel set of effective tools to assess the robustness of social systems against malicious attacks and assist policy makers in designing policies to promote social change.


The promising results presented in this paper outline several lines of future research. First, while this paper focuses on static control policies, one could conjecture that, after a critical mass is reached, one may uplift the control action. Future research should focus on designing dynamic control policies, towards minimizing the total control effort. 
Second, this paper focuses on the problem of steering a population to a consensus. In some applications, however, policy makers may want to favor diversity by reaching a non-consensus state. Extending our methodology to different control objectives is thus a key objective of our future research. Third, in this paper, we have focused on minimizing the number of controlled nodes. Generalizing \cite{raineri2024}, this paper already differentiates between the possibilities of controlling only actions or only opinions. However, to further extend the applicability of the model, different nodes may be associated with different control costs, as well as the fact that the cost may differ between opinion and action control. Hence, a natural extension of our control problem involves generalizing our theoretical tools to optimize non-trivial control costs functions. 
Fourth, the improved performance of Algorithm~\ref{alg:optimal_C} in identifying the optimal control set compared to classical heuristics suggests that our approach can be possibly extended to similar control problems for other multi-dimensional supermodular games. 

\section*{Acknowledgments}

The authors are indebted to Giacomo Como and Fabio Fagnani for precious discussion.

\appendix

\subsection{Proof of Proposition~\ref{prop:super}}\label{sec:proof_super}
Being the domain compact and the utility function upper-semicontinuous, to prove supermodularity, we need to check that \eqref{eq:utility} has increasing differences~\cite{topkis1998book}, i.e., that given $\vect{z_i}'\geq \vect{z_i}$ and $\vect{z_{-i}}'\geq \vect{z_{-i}}$, it holds that $$\Delta_i(\vect{z_i}',\vect{z_i},\vect{z'}):=u_i(\vect{z_i}',\vect{z'})-u_i(\vect{z_i},\vect{z'})\geq \Delta_i(\vect{z_i}',\vect{z_i},\vect{z_{-i}}).$$ Using \eqref{eq:utility}, we compute    \begin{equation}\label{eq:delta}\begin{array}{l}
\Delta_i(\vect{z_i}',\vect{z_i},\vect{z_{-i}})={\lambda_i(1-\beta_i)}\sum_{j\in\mathcal V}a_{ij}(x_i'-x_i)x_j\\\quad+2\beta_i(1-\lambda_i)\sum_{j\in\mathcal V}w_{ij}y_j(y_i'-y_i)+\psi(\vect{z_i}',\vect{z_i}),
    \end{array}\end{equation}
    where $\psi(\vect{z_i}',\vect{z_i})$ is a function that depends only on $\vect{z_i}'$ and $\vect{z_i}$. Hence, being $x_i'\geq x_i$ and $y_i'\geq y_i$, \eqref{eq:delta} is monotonic nondecreasing in $x_j$ and $y_j$, $j\in\mathcal V$. Hence, $\Delta_i(\vect{z_i}',\vect{z_i},\vect{z_{-i}}')\geq \Delta_i(\vect{z_i}',\vect{z_i},\vect{z_{-i}})$ for any $\vect{z_{-i}}'\geq \vect{z_{-i}}$, yielding the claim.\qed

\subsection{Proof of Theorem~\ref{th:convergence}}\label{sec:proof_convergence}
Consider the controlled game under Assumption~\ref{a:initial_condition}. 
 Proposition~\ref{prop:super} guarantees that the game is supermodular. Moreover, we observe that under Assumption~\ref{a:initial_condition}, \eqref{eq:br}, as previously introduced, can be equivalently characterized as the minimal best response $(x_i(t+1),y_i(t+1))=\min(\argmax u_i(\vect{z_i},\vect{z_{-i}}))$~\cite{Milgrom1990}. It is known from~\cite{Milgrom1990} that the trajectory of a minimal best response with  smallest initial conditions is monotonically non-decreasing~\cite{Milgrom1990}. Hence, since Assumption~\ref{a:initial_condition} fixes the initial condition for all the uncontrolled nodes to $(-1, -1)$, then for all trajectories, it holds $\vect{x}(t)\geq \vect{x}(t-1)$ and $\vect{y}(t)\geq \vect{y}(t-1)$ for all $t\geq 0$. Note that the same monotonicity results can be obtained by using the explicit dynamics in Proposition~\ref{prop:dynamics}, following the arguments used for the special case of joint control in~\cite[Lemmas~1--2]{raineri2024}. 
 Finally, monotonicity implies convergence due to the monotone convergence theorem~\cite{bartle1976}. \qed

\subsection{Proof of Lemma~\ref{lemma:intermediate  C}}\label{sec:proof_intermediate}
To prove the first claim, we observe from \eqref{eq:delta} that $\delta_i(\vect z)$ is a monotonic nondecreasing function of $\vect{z}$. Consequently, from \eqref{eq:dinamics}, $\vect{x}(t+1)$ and $\vect{y}(t+1)$ are monotonic nondecreasing functions of $\vect{z}(t)$. Let us define the controlled coevolutionary dynamics $\vect{z}(t)$, with initial condition $\vect{z}(0)$ according to \eqref{eq:control_set_def}. Let us consider control sets $\mathcal {\bar C}^X\supset \mathcal C^X$ and $\mathcal {\bar C}^Y\supset \mathcal C^Y$, and corresponding dynamics $\vect{\bar z}(t)$ with initial conditions defined according to according to \eqref{eq:control_set_def}. Then, fixed any common activation sequence for $\vect{z}(t)$ and $\vect{\bar z}(t)$, the monotonicity properties described in the above yields $\vect{\bar z}(t)\geq \vect{z}(t)$. Finally, for almost every activation sequence we have that $\vect{x}(t)\to\vect{1}$, being $\phi(\mathcal C^X,\mathcal C^Y)=1$. Consequently, also $\vect{\bar x}(t)\to\vect{1}$ for all configurations (except for a set of measure zero), yielding the claim. The second claim is proved following a similar (symmetric) argument. \qed

\subsection{Proof of Corollary~\ref{cor:joint}}\label{sec:proof_joint}
  Let $\mathcal C^X$, $\mathcal C^Y$ be an optimal solution of \eqref{eq:minimal}. Then, we define $\mathcal C:=\mathcal C^X\cup\mathcal C^Y$, $\mathcal {\hat C}^X:=\mathcal C\cap \mathcal V^X$, and $\mathcal {\hat C}^Y:=\mathcal C\cap \mathcal V^Y$. Clearly, $\mathcal {\hat C}^X\supseteq \mathcal {C}^X$ and $\mathcal {\hat C}^Y\supseteq \mathcal {C}^Y$. Hence, by Lemma~\ref{lemma:intermediate  C}, since $\phi(\mathcal C^X,\mathcal C^Y)=1$, also $\phi(\mathcal {\hat C}^X,\mathcal {\hat C}^Y)=1$. Moreover,  $\mathcal C^X\subseteq \mathcal V^X$ and $\mathcal C^Y\subseteq \mathcal V^Y$ by construction, so $(\mathcal {\hat C}^X,\mathcal {\hat C}^Y)$ is a feasible solution of \eqref{eq:minimal}. Finally, we observe that $\mathcal {\hat C}^X\cup\mathcal {\hat C}^Y=\mathcal C$. Hence, $|\mathcal {\hat C}^X\cup\mathcal {\hat C}^Y|=|\mathcal {C}^X\cup\mathcal {C}^Y|=|\mathcal C|$, and $(\mathcal {\hat C}^X,\mathcal {\hat C}^Y)$ is an optimal solution of \eqref{eq:minimal}, yielding the claim.\qed

\subsection{Proof of Theorem~\ref{th:np}}\label{sec:proof_np}
First, we recall that when $\lambda_i=1$, for all $i\in\mathcal V$, our model reduces to a majority game, for which it is well-known that finding the minimal control set is NP-hard~\cite{Como2022supermodular}. Hence, also Problem~\ref{problem2} is NP-hard. We now show that it belongs to the NP-class, which yields NP-completeness. 
To prove that Problem~\ref{problem2} is NP, we demonstrate that, given an instance of the coevolutionary dynamics and a control set $(\mathcal C^{X}, \mathcal C^{Y})$, we can check whether $(\mathcal C^{X}, \mathcal C^{Y})$ is a feasible solution of \eqref{eq:minimal} in a polynomial time. In other words, if there exists an algorithm able to solve Problem~\ref{problem} in polynomial time, then clearly one can determine whether $\phi(\mathcal C^{X}, \mathcal C^{Y})=1$, while checking whether $\mathcal C^{X}$ is a subset of $\mathcal V^{X}$ and $\mathcal C^{Y}$ of $\mathcal V^{Y}$ can be checked in time $O(n)$ using, e.g., an hash function. To prove that there exists an algorithm to solve Problem~\ref{problem} in polynomial time, we refer to Theorem~\ref{th:algorithm}, which proves that Algorithm~\ref{alg} solves  Problem~\ref{problem} in time $O(n^3)$, yielding the claim. \qed

\subsection{Proof of Proposition~\ref{prop:sub}}\label{sec:proof_sub}
We build a counterexample. Consider $2$ nodes connected by a link with $w_{12}=a_{12}=w_{21}=a_{21}=1/3$ and $w_{11}=a_{11}=w_{22}=a_{22}=2/3$, and $\mathcal R(t)=\mathcal V$, for all $t$. We consider $\mathcal C^{X}_1=\{1\}$, $\mathcal C^{X}_2=\{2\}$, and $\mathcal C^{Y}=\mathcal V$.
Clearly, it holds $\mathcal C^{X}_1\cup \mathcal C^{X}_2=\mathcal V$, which implies that $\phi(\mathcal C^{X}_1\cup\mathcal C^{X}_2,\mathcal C^{Y})=\phi(\mathcal V,\mathcal V)=1$. On the other hand, for control sets $\mathcal C^{X}_1,\mathcal C^{Y})$, we immediately observe that the only state that can change is $x_2(t)$. At the first time instant $t$ at which $2\in\mathcal R(t)$, which occurs in finite time due to Assumption~\ref{a:activation}, individual $2$ switches to $+1$ iff $\delta_2(\vect z(t)) = 2 \beta_2 (1-\lambda_2) - (1- 
 \beta_2)>0$, according to Proposition~\ref{prop:dynamics}. By symmetry, a similar condition holds for individual $1$ when $\mathcal C^{X}_2$, exchanging the role of individuals $1$ and $2$. This implies that
$$\phi(\mathcal C^{X}_i,\mathcal C^{Y})=\left\{\begin{array}{ll}1&\text{if }\lambda_{3-i}\leq\frac{3\beta_{3-i}-1}{2\beta_{3-i}},\\0&\text{otherwise}.\end{array}\right.
$$
Hence, if we set $\beta_1=\beta_2=1/2$ and $\lambda=2/3$, then $\phi(\mathcal C^{X}_1,\mathcal C^{Y})=\phi(\mathcal C^{X}_2,\mathcal C^{Y})=0$, which implies that the condition $\phi(\mathcal C^{X}_1,\mathcal C^{Y})+\phi(\mathcal C^{X}_2,\mathcal C^{Y})\geq  \phi(\mathcal S\cup \mathcal T,\mathcal C^{Y})+ \phi(\mathcal S\cap \mathcal T,\mathcal C^{Y})$ required by submodularity~\cite{topkis1998book} is not satisfied by the first variable. Similar, we can build a counterexample to prove that the function is not submodular also with respect to the second variable, yielding the claim.\qed

\subsection{Proof of Theorem~\ref{th:algorithm}}\label{sec:proof2}

We start by proving the following result.

\begin{lemma}
\label{cor:delta_neg}
Given control sets $(\mathcal C^{X},\mathcal C^{Y})$ and a vector $\vect{\hat x}\in\{-1,1\}^n$, let $\vect{\hat y}$ be the solution of \eqref{y-equilibrium} given  $\vect{\hat x}$. 
Then, configuration $\vect{\hat z}=(\vect{\hat x}, \vect{\hat y})$ is an equilibrium for the controlled coevolutionary dynamics  under Assumptions~\ref{a:activation}--\ref{a:initial_condition} iff 
 $\hat x_i \delta_i(\vect{\hat z})\geq 0$, $\forall\,i \notin \mathcal C^X$. If there exists at least an individual $i\notin\mathcal C^X$ with $\hat x_i \delta_i(\vect{\hat z})< 0$, there exist no  equilibria with $\vect{x}=\vect{\hat x}$. 
\end{lemma}
\begin{proof}
Fixed the action vector $\vect{\hat x}$ and  the opinion of $j\in\mathcal C^{Y}$ to $\hat y_j=1$ (because of Assumption~\ref{a:initial_condition}), the  dynamics in \eqref{y-dinamic} for a generic individual $i\notin \mathcal C^Y$ reduces to 
$$
 \begin{array}{ll}
    y_i(t+1)=\displaystyle(1-\lambda_i)\Big[\sum_{j \notin \mathcal C^Y}{w}_{ij}y_j(t)+\sum_{j \in \mathcal C^{Y}}{w}_{ij}\hat y_i\Big] + \lambda_i \hat x_i\\\qquad\qquad=\displaystyle(1-\tau_i)\sum\nolimits_{j \notin \mathcal C^Y} \bar {w}_{ij}y_j(t) + \tau_i u_i,
 \end{array}$$
 with $$\tau_i=1-(1-\lambda_i)\sum\nolimits_{j\notin\mathcal C^Y}w_{ij},$$ $$\bar w_{ij}=\frac{w_{ij}}{\sum_{k\notin\mathcal C^Y}w_{ik}},$$ and $$u_i=\frac{(1-\lambda_i)(1-\sum_{j\notin \mathcal C^Y}w_{ij})+\lambda_i\hat x_i}{1-(1-\lambda_i)\sum_{j\notin \mathcal C^Y}w_{ij}},$$
 with the convention that, $\bar w_{ij}=0$ if $w_{ij}=0$.  This can be seen as the update rule of a Friedkin--Johnsen opinion dynamics model~\cite{proskurnikov2017tutorial}, which is known to converge under Assumption~\ref{a:activation} to the unique solution of $$\hat y_i=(1-\tau_i)\sum\nolimits_{j \notin \mathcal C^Y} \bar {w}_{ij}\hat y_j + \tau_i u_i,$$ which coincides with \eqref{y-equilibrium}. See,~\cite{proskurnikov2018tutorial_2} for more details. Hence, $\vect{\hat z} = (\vect{\hat x}, \vect{\hat y})$ is the only admissible equilibrium with action vector equal to $\vect{\hat x}$.
 
 Now, we observe that $\vect{\hat z}$ is an equilibrium iff there are no individuals that would change their action according to \eqref{x-dinamic} when the system is at $\vect{\hat z}$. This corresponds to verify that all individuals $i \notin \mathcal C^{X}$ with $\hat x_i = -1$ have $\delta_i(\vect{\hat z})\leq0$, and all those 
 with $\hat x_i = 1$ have $\delta_i(\vect{\hat z})\geq0$. In fact, if there exists $i\notin \mathcal C^X$ with $\hat x_i=-1$ and $\delta_i(\vect{\hat z}) >0$, then Assumption~\ref{a:activation} guarantees that within a finite time-window, $i$ activates and flips action to $+1$ (being $\delta_i(\vect{\hat z}) >0$).
\end{proof}

Now, we use Lemma~\ref{cor:delta_neg} to prove the following result.

\begin{lemma}
\label{lem:equilibrium}
The equilibrium reached by a controlled coevolutionary dynamics that satisfies Assumptions~\ref{a:activation} and~\ref{a:initial_condition} with control sets $(\mathcal C^{X},\mathcal C^{Y})$ is $(\vect{x}^*,\vect y^{*})$, with $\vect{x^*}$ defined as in \eqref{eq:x-candidate} with $\mathcal A(k)=\mathcal A_f$ (output of Algorithm~\ref{alg}) and $\vect{y}^*$ solution of \eqref{y-equilibrium} given $\vect{x}^*$.
\end{lemma}
\begin{proof}
In the first iteration of the algorithm ($k=1$), Lemma~\ref{cor:delta_neg} establishes that state $\vect{\hat z}$ defined using  \eqref{eq:x-candidate} with set $\mathcal A(1)$ and \eqref{y-equilibrium} is an equilibrium iff $\mathcal A(2)=\mathcal A(1)$. Otherwise, we will now prove that individuals in $\mathcal A(2)\setminus\mathcal A(1)$ will eventually switch action to $+1$. 
In fact, as long as $\vect{x}(t)=\vect{\hat x}$, then $y_j(t)$ converges asymptotically to $\hat y_j$ for all $j\notin\mathcal C^Y$ (due to the observations made in the proof of Lemma~\ref{cor:delta_neg}). Hence, $\delta_i(\vect{z}(t))$ converges asymptotically to $\delta_i(\vect{\hat z})>0$. By continuity, $\exists\,\tilde t$ such that $\delta_i(\vect{z}(t))>0$ for all $t\geq \tilde t$, as long as $\vect{x}(t)=\vect{\hat x}$. Moreover, since $\delta_i(\vect{z}(t))$ is monotonically increasing in $\vect{z}$ and $\vect{z}(t)$ is monotonically increasing in $t$, then $\delta_i(\vect{z}(t))$ is a monotonically increasing function of time. This implies that $\delta_i(\vect{z}(t))>0$ for all $t\geq \tilde t$. This, together with Assumption~\ref{a:activation}, guarantees that $i$ switches to $+1$ (Proposition~\ref{prop:dynamics}) and cannot switch back (Theorem~\ref{th:convergence}), then $x_i(t)=+1$ for all $t\geq\tilde t+T$.

If $\mathcal A(2)=\mathcal A(1)$, then the system necessarily converges to the equilibrium $\vect{\hat z}$, yielding the claim. Otherwise,  $\vect{\hat z}$ is not an equilibrium. In this case, all individuals in $\mathcal A(2)\setminus\mathcal A(1)$ will necessarily switch action to $+1$ in finite time. Hence, we re-iterate considering the set $\mathcal A(2)$ and computing the corresponding $\vect{\vect{\hat z}}$, observing that, if $i$ has switched to $+1$, Theorem~\ref{th:convergence} guarantees that $i$ will never switch back, so we just need to check whether all $i\in\mathcal V\setminus\mathcal A(2)$ have $\delta_i(\hat z)\leq 0$ to get the terminal condition $\mathcal A(3)=\mathcal A(2)$, for which the system necessarily converges to the equilibrium $\vect{\vect{\hat z}}$. Otherwise, we re-iterate the process. Finally, in each iteration   $k$ in which the terminal condition is not met, the size of $\mathcal A(k)$ increases by at least $1$, implying that within at most $k=n-|\mathcal C^{X}|$ iterations we would get  $\mathcal A(k)=\mathcal V$, for which $\vect{\hat z}=(\vect{1},\vect{1})$ is a trivial equilibrium, terminating the algorithm.
\end{proof}

Theorem~\ref{th:convergence} implies that a controlled coevolutionary dynamics always converge to an equilibrium. Lemma~\ref{lem:equilibrium} implies that the equilibrium is independent of the activation sequence, but depends only on the model parameters and on the initial condition, which are  determined by $\mathcal C^{X}$ and $\mathcal C^{Y}$. Hence, fixed the parameters and given the $\mathcal C^{X}$ and $\mathcal C^{Y}$ either the system converges to $\vect{x}=\vect{1}$, implying $\phi(\mathcal C^{X},\mathcal C^{Y})=1$ or to any other equilibrium, implying $\phi(\mathcal C^{X},\mathcal C^{Y})=0$.

Finally, observe that \eqref{y-equilibrium} can be rewritten as $\vect{\hat{y}} = (\mat I - (\mat I-\text{diag}(\vect \lambda))\mat W)^{-1} \text{diag}(\vect \lambda)\vect{\hat{x}}$. The matrix $\mat M:=(\mat I - [(\vect 1-\vect \lambda)]\mat W)^{-1}$ does not depend on the iteration step, thus it can be computed once at the beginning of the iterations (such procedure requires $O(n^3)$ operations). Then, at each iteration of Algorithm~\ref{alg}, the dominant operation is the computation of $\vect{\hat{y}}$ which requires $O(n^2)$ operations. Since the number of iterations is at most $n-|\mathcal C^{X}|$, the total computational complexity of Algorithm~\ref{alg} is $O(n^3)$, yielding the claim. \qed

\end{document}